\documentclass[10pt, twocolumn]{IEEEtran}

\usepackage{psfrag}
\usepackage{amssymb}
\usepackage{amsmath}
\usepackage{pifont}
\usepackage{cite} 
\usepackage{graphics}
\usepackage{graphicx}
\usepackage{epsfig}
\usepackage{subfigure}
\usepackage{url}
\usepackage{amscd}
\usepackage{threeparttable}
\usepackage{enumerate}
\usepackage[colorlinks, citecolor=blue,linkcolor=blue]{hyperref}

\newtheorem{theorem}{Theorem}
\newtheorem{remark}{Remark}
\newtheorem{lemma}{Lemma}

\newtheorem{proposition}{Proposition}

\begin{document}

\title{Exact Recovery of Sparse Signals Using Orthogonal Matching Pursuit: How Many Iterations Do We Need?}

\author{\IEEEauthorblockN{Jian Wang$^{\dag}$ and Byonghyo Shim$^{\ddag}$} \\
\IEEEauthorblockA{$^\dag$Department of Electrical \& Computer
Engineering, Duke University \\
$^\ddag$Department of Electrical \& Computer Engineering, Seoul National University \\
Email: {jian.wang@duke.edu, bshim@snu.ac.kr}} 
}

\maketitle

\begin{abstract}
Orthogonal matching pursuit (OMP) is a greedy algorithm widely used
for the recovery of sparse signals from compressed measurements. In
this paper, we analyze the number of iterations required for the OMP algorithm to perform exact recovery of sparse signals. Our analysis shows that OMP can accurately recover all $K$-sparse signals within $\lceil 2.8 K \rceil$ iterations when the measurement matrix satisfies a restricted isometry property (RIP).  Our result improves upon the recent result of Zhang and also bridges the gap between Zhang's result and the fundamental
limit of OMP at which exact recovery of $K$-sparse signals cannot be uniformly guaranteed.
\end{abstract}

\section{Introduction} \label{sec:I}
Recently, there has been a growing interest in recovering sparse
signals from compressed
measurements~\cite{donoho2006compressed,candes2006robust,candes2006near,needell2009cosamp,dai2009subspace,foucart2011hard,Soussen2013joint,
chen2013oracle}.
The main goal of this task is to accurately estimate a high
dimensional $K$-sparse vector $\mathbf{x} \in {\mathbb{R}^n}$
($\|\mathbf{x}\|_0 \leq K$) from a small number of linear
measurements $\mathbf{y} \in {\mathbb{R}^m}$ ($m \ll n$). The
relationship between the signal vector and the measurements is given
by
\begin{equation}
\mathbf{y} = \mathbf{\Phi x},
\end{equation}
where ${\mathbf{\Phi} \in {\mathbb{R}^{m \times n}} }$ is often
called the measurement (sensing) matrix.
%
%
A key finding in the sparse recovery problem is that one can recover
the original vector $\mathbf{x}$ with far fewer measurements than
traditional approaches use, as long as the signal to be recovered is
sparse and the measurement matrix roughly preserves the energy of
the signal of interest. Among many algorithms designed to recover
the sparse signal, orthogonal matching pursuit (OMP) algorithm has
received much attention for its competitive performance as well as
practical benefits, such as implementation simplicity and low
computational complexity~\cite{pati1993orthogonal,tropp2004greed}.
In essence, the OMP algorithm estimates the input sparse vector
$\mathbf{x}$ and its support (index set of nonzero elements) in an
iterative fashion, generating a series of locally optimal updates
fitting the measurement data. Specifically, at each iteration the
index of column that is mostly correlated with the modified
measurements (often called residual) is chosen as a new element of
the estimated support set. The vestiges of columns in the estimated
support are then eliminated from the measurements, yielding a new
residual for the next iteration. See Table~\ref{tab:OMP} for a detailed
description of the OMP algorithm.

%
%

\setlength{\arrayrulewidth}{1.6pt}
\begin{table}
  \centering
\caption{The OMP Algorithm} \label{tab:OMP}
\vspace{-2mm}
\begin{tabular}{@{}ll}
\hline \\ \vspace{-12pt} \\
\textbf{~~~Input}       &$\mathbf{\Phi}$, $\mathbf{y}$, and maximum iteration number $k_{\max}$.~ \\
\textbf{~~~Initialize}  & iteration counter $k = 0$, \\
                     & estimated support ${T}^{0} = \emptyset$, \\
                     & and residual vector $\mathbf{r}^{0} = \mathbf{y}$.
                     \\
\textbf{~~~While}       & $k < k_{\max}$, \textbf{do}\\
                     & $k = k + 1$. \\
                     & Identify \hspace{2.5mm}${t}^{k}  = \underset{i \in \Omega \backslash {T}^{k - 1}}{\arg \max} |\langle \phi_i, \mathbf{r}^{k - 1} \rangle|$. \\
                     & Enlarge \hspace{1.65mm}${T}^{k} = {T}^{k - 1} \cup t^{k}$. \\
                     & Estimate \hspace{0.85mm}$\hat{\mathbf{x}}^{k} = \underset{\mathbf{u}:\textit{supp}(\mathbf{u}) = {T}^k}{\arg \min} \|\mathbf{y}-\mathbf{\Phi} \mathbf{u}\|_2$. \\
                     & Update \hspace{2.99mm}$\mathbf{r}^{k} = \mathbf{y} - \mathbf{\Phi} \hat{\mathbf{x}}^{k}$. \\
\textbf{~~~End}         \\
\textbf{~~~Output}  &$T^k$ and $\hat{\mathbf{ x}}^k$. \\
\vspace{-3pt} \\
\hline
\end{tabular}
\end{table}
\setlength{\arrayrulewidth}{1.6pt}

Over the years, the OMP algorithm has long been considered as a heuristic algorithm hard to be analyzed. Recently, however, much research has been devoted to discovering the condition of OMP ensuring exact recovery of sparse signals.
In one direction, studies to identify the recovery condition using
probabilistic analyses have been proposed. Tropp and Gilbert
showed that when the measurement matrix $\mathbf{\Phi}$ is generated \text{i.i.d.} at random, and the measurement size is on the order of $K \log n$, OMP ensures the accurate recovery of every fixed $K$-sparse signal with overwhelming probability~\cite{tropp2007signal}.
%
%
Another line of work is to characterize exact recovery conditions
of OMP using properties of measurement matrices, such as the mutual
incoherence property (MIP)~\nocite{donoho2001uncertainty} and the restricted isometry property (RIP)~\cite{candes2005decoding}.
A measurement matrix $\mathbf{\Phi}$ is said to satisfy the RIP of
order $K$ if there exists a constant $\delta (\mathbf{\Phi}) \in (0, 1)$ such
that 
\begin{equation} \label{eq:RIP}
\left( {1 - \delta(\mathbf{\Phi})} \right)\| {\mathbf{x}} \|_2^2
\leq \| {{\mathbf{\Phi x}}} \|_2^2 \leq \left( {1 +
\delta(\mathbf{\Phi})} \right)\| {\mathbf{x}} \|_2^2
\end{equation}
for any $K$-sparse vector $\mathbf{x}$. In particular, the minimum
of all constants $\delta (\mathbf{\Phi})$ satisfying (\ref{eq:RIP})
is called the restricted isometry constant (RIC) and denoted by
$\delta_K (\mathbf{\Phi})$. In the sequel, we use $\delta_K$ instead of
$\delta_K (\mathbf{\Phi})$ for brevity.
In~\cite{davenport2010analysis}, Davenport and Wakin showed that OMP
ensures exact reconstruction of any $K$-sparse signal under
\begin{equation}
\delta _{K + 1} < \frac{1}{3\sqrt K }.
\end{equation}
Since then, many efforts have been made
to improve this condition~\cite{mo2012remarks,wang2012Recovery,wen2013improved,  
yang2014new,chang2014improved}.  
Recently, Mo has improved the condition to~\cite{mo2015sharp}
\begin{equation}
\delta _{K + 1}  <
\frac{1}{\sqrt{K + 1}}, \label{eq:monew}
\end{equation}
which is in fact sharp since there exist measurement matrices $\mathbf{\Phi}$ with $\delta_{K + 1} = \frac{1}{\sqrt{K + 1}}$, for which OMP fails to recover
the original $K$-sparse signal from its measurements in $K$ iterations~\cite{wen2013improved,mo2015sharp}. Therefore, in order to uniformly recover all $K$-sparse signals in $K$ iterations of OMP, the RIC should at least be inversely proportional to $\sqrt K$.

While aforementioned studies of OMP have focused on the scenario
where the number of iterations is limited to the sparsity $K$, there
have been recent works investigating the behavior of OMP when it
performs more than $K$ iterations~\cite{zhang2011sparse,livshits2012efficiency,foucart2013stability,livshitz2014sparse}
or when it chooses more than one index per iteration~\cite{liu2012orthogonal,huang2011recovery,wang2012Generalized}. Both in theoretical
performance guarantees and empirical simulations, these approaches
provide better results and also offers new insights into this
seemingly simple-minded yet clever algorithm.
Livshitz showed that with proper choices of $\alpha$ and $\beta$
($\alpha \sim 2 \cdot 10^{5}$ and $\beta \sim 10^{-6}$), OMP
accurately reconstructs $K$-sparse signals in $\left\lfloor \alpha
K^{1.2}\right\rfloor$ iterations under~\cite{livshits2012efficiency}
\begin{equation}
\delta_{\alpha K^{1.2}} =
\beta K^{-0.2}.
\end{equation}
Although the RIC decays slowly with $K$ (when compared to
that in the results of~\cite{davenport2010analysis,mo2012remarks,wang2012Recovery,wen2013improved,
yang2014new, chang2014improved,mo2015sharp}),
and thus offers significant benefits in the measurement size, it is not easy to
enjoy the benefits in practice, since it requires too many iterations.
Recently, it has been shown by Zhang that OMP recovers any $K$-sparse signal with $30K$ iterations under~\cite{zhang2011sparse}
\begin{equation}
 \delta_{31K} < \frac{1}{3}.
 \end{equation} The significance of this result is that when running $30K$ iterations, OMP can recover $K$-sparse signals accurately with the RIC being an absolute constant independent of $K$, exhibiting the reconstruction capability comparable to the state of the art sparse recovery algorithms (e.g., Basis Pursuit~\cite{chen2001atomic} and CoSaMP~\cite{needell2009cosamp}).
In the sequel, to distinguish the OMP algorithm running $\lceil cK \rceil$ ($c > 1$) iterations
from the conventional OMP algorithm running $K$ iterations, we denote it as
OMP$_{cK}$.

In this paper, we go further to investigate how many iterations of OMP$_{cK}$ would be enough to guarantee exact recovery of sparse signals, given that the RIC is an absolute constant. 
Note that running fewer number of iterations offers many computational benefits in practice. 
Our main result, described in Theorem~\ref{thm:general_3}, is that OMP$_{cK}$
accurately recovers all $K$-sparse signals $\mathbf{x}$ from the measurements
$\mathbf{y} = \mathbf{\Phi x}$ if $c$ satisfies a condition
expressed in terms of the RIC.
The main significance of this result is that as long as 
\begin{equation}
 c > 4 \log 2 \approx 2.8,
 \end{equation} there always exists an RIC, which is an absolute constant, such that the underlying condition is fulfilled. This means that the required number of iterations of OMP for exact recovery of sparse signals can be as few as $\lceil 2.8K \rceil$. 

It is well known that with overwhelming probability, random matrices (e.g., random Gaussian, Bernoulli, and partial Fourier matrices) satisfy the RIP when the number of measurements scales nearly linearly with the sparsity $K$~\cite{candes2005decoding,baraniuk2008simple}. In view of this, our result implies that if a $K$-sparse signal is measured by random matrices, it can be recovered with the nearly optimal number of measurements via the OMP algorithm running only $\lceil 2.8K \rceil$ iterations.

We summarize notations used throughout this paper. $\Omega = \{1,2,
\cdots, n\}$. $T = {supp}(\mathbf{x})= \{i |i \in \Omega, x_i\neq
0\}$ is the set of nonzero positions in $\mathbf{x}$. For $S
\subseteq \Omega$, $T \backslash S$
is the set of all elements contained in $T$ but not in $S$. $|S|$ is the cardinality of $S$. 
If $|S| \neq 0$, ${{\mathbf{x}}_S} \in \mathbb{R}^{|S|}$ is the restriction of the vector $\mathbf{x}$ to those elements indexed by $S$. Similarly,
${{\mathbf{\Phi }}_S} \in {\mathbb{R}^{m \times \left| S \right|}}$
is a submatrix of ${\mathbf{\Phi }}$ that only contains columns
indexed by $S$. If $\mathbf{\Phi}_S$
has full column rank, $\mathbf{\Phi}_S^\dagger =
(\mathbf{\Phi}'_S\mathbf{\Phi}_S)^{-1}\mathbf{\Phi}'_S$ is the
Moore-Penrose pseudoinverse of $\mathbf{\Phi}_S$ where $\mathbf{\Phi}_{S}'$ is the transpose of the matrix
$\mathbf{\Phi}_{S}$. $\mathcal{P}_{S}=\mathbf{\Phi}_{S}
\mathbf{\Phi}_{S}^\dagger$ is the projection onto
${span}(\mathbf{\Phi}_{S})$ (i.e., the span of columns in
$\mathbf{\Phi}_S$). In particular, if $S = \emptyset$,  $\mathbf{x}_\emptyset$ is a $0$-by-$1$ empty vector with $\ell_2$-norm $\|\mathbf{x}_{\emptyset}\|_2 = 0$, ${{\mathbf{\Phi}}_\emptyset}$ is an $m$-by-$0$ empty matrix, and ${{\mathbf{\Phi}}_\emptyset} \mathbf{x}_\emptyset$ is an $m$-by-$1$ zero matrix~\cite{{bernstein2009matrix}}.

%
%

%
%

\section{Exact Recovery of Sparse Signals via OMP}
\label{sec:cKOMP}
Suppose $\lceil cK \rceil$ ($c > 1$) is the number of
iterations ensuring selection of all support indices of the
$K$-sparse signal $\mathbf{x}$ (i.e., $T\subseteq T^{\lceil cK
\rceil}$).
Then the estimated support set $T^{\lceil cK \rceil}$ may contain
indices not in $T$. Even in this situation, the final result is
unaffected and the original signal $\mathbf{x}$ is recovered
accurately (i.e., $\hat{\mathbf{x}}^{\lceil cK \rceil}
= \mathbf{x}$) because 
\begin{eqnarray} \label{eq:LS}
 \hat{\mathbf{x}}^{\lceil cK \rceil} = \arg \min_{\mathbf{u}: {supp}(\mathbf{u}) = T^{\lceil cK \rceil}} {\|\mathbf{y}-\mathbf{\Phi}\mathbf{u}\|}_{2} \nonumber
\end{eqnarray}
and
\begin{eqnarray}
(\hat{\mathbf{x}}^{\lceil cK \rceil})_{T^{\lceil cK \rceil}} &=& \label{eq:bff1} \mathbf{\Phi}^\dag_{T^{\lceil cK \rceil}}  \mathbf{y} ~=~ \label{eq:bff1} \mathbf{\Phi}^\dag_{T^{\lceil cK \rceil}}  \mathbf{\Phi}_{T} \mathbf{x}_{T} \nonumber
 \\
 &\overset{(a)}{=}& \mathbf{\Phi}^\dag_{T^{\lceil cK \rceil}} \mathbf{\Phi}_{T^{\lceil cK \rceil}} \mathbf{x}_{T^{\lceil cK \rceil}} \nonumber
 \\
 &=& \label{eq:bff3} \mathbf{x}_{T^{\lceil cK \rceil}},
\end{eqnarray}
where (a) is
from the fact that $\mathbf{x}$ is supported on $T$ and hence
$\mathbf{x}_{T^{\lceil c K \rceil} \backslash T} = \mathbf{0}$.
This simple property allows us to investigate OMP running more than
$K$ iterations.
While running more iterations than the sparsity level would be
beneficial in obtaining better recovery bound, at the same time it
induces additional computational burden.
In fact, since the dimension of the matrix to be inverted increases
by one per iteration (see Table~\ref{tab:OMP}), both the operation cost
and running time increase cubically with the number of
iterations.
Therefore, it is of importance to investigate the lower bound for the number of
iterations ensuring accurate identification of the whole support
(i.e., the lower bound for $c$ that ensures $T\subseteq T^{\lceil cK
\rceil}$). 
Our result is described in the following theorem.
 
\vspace{1mm} 
%
%
\begin{theorem}\label{thm:general_1}
Let $\mathbf{x} \in \mathbb{R}^n$ be any $K$-sparse signal supported
on $T$ and $\mathbf{\Phi} \in \mathbb{R}^{m \times n}$ be the
measurement matrix. Furthermore, let $N^{k} := |T \backslash T^{k}|$ be
the number of remaining support indices after $k$ ($0 \leq k \leq
\lceil c K \rceil$) iterations of OMP$_{cK}$. Then if
$\mathbf{\Phi}$ obeys the RIP of order $s := |T \cup T^{k + \lfloor
c N^{k} \rfloor}|$ and
\begin{equation}
    c \geq - \frac{4 (1 + \delta_1)}{1 - \delta_{s}}
  \log\left(\frac{1}{2} - \frac{1}{2}\sqrt{\frac{\delta_{N^{k}} + \delta_{s}}{1 + \delta_{N^{k}}}}
  \right),
\label{eq:j2jiayou111}
\end{equation}
then OMP$_{cK}$ satisfies $T \subseteq T^{k + \lceil c N^{k} \rceil}$.
\end{theorem}

The proof of Theorem~\ref{thm:general_1} will be given in Section
\ref{sec:proof}. The key point of Theorem \ref {thm:general_1} is
that after performing $k$ ($0 \leq k \leq
\lceil c K \rceil$) iterations, OMP$_{cK}$ selects the remaining $N^k$
support indices within $\lceil c N^{k} \rceil$ additional iterations,
as long as the condition in~\eqref{eq:j2jiayou111} is fulfilled. In
particular, when $k = 0$, $N^k = |T \backslash T^k| = |T \backslash
\emptyset| = K$ and $T \subseteq T^{\lceil c K \rceil}$ holds true
under
\begin{equation}
    c \geq \hspace{-.5mm} - \frac{4 (1 + \delta_1)}{1 - \delta_{|T \cup T^{\lfloor cK \rfloor}|}}
  \log\left(\hspace{-.5mm} \frac{1}{2} \hspace{-.5mm} - \hspace{-.5mm} \frac{1}{2}\sqrt{\frac{\delta_{K} \hspace{-.5mm} + \delta_{|T \cup T^{\lfloor cK \rfloor}|}}{1 + \delta_{K}}}
  \right)\hspace{-1mm}.
\label{eq:j2jiayou1116181}
\end{equation}
Further, from monotonicity of the RIC (i.e., $\delta_{K_1} \leq \delta_{K_2}$ for $K_1 \leq K_2$), one can easily show that \eqref{eq:j2jiayou1116181} is guaranteed by 
\begin{eqnarray}
    c \hspace{-.25mm}  \geq \hspace{-.5mm} - \hspace{-.5mm} \frac{4 (1 \hspace{-.25mm} + \hspace{-.25mm} \delta)}{1 - \delta}
  \log \hspace{-.25mm} \left(\hspace{-.5mm} \frac{1}{2} \hspace{-.5mm} - \hspace{-.5mm} \sqrt{\frac{\delta}{2 \hspace{-.25mm} + \hspace{-.25mm} 2 \delta}}
  \right)~\hspace{-1mm}\text{where}~\delta \hspace{-.5mm}:= \delta_{\lfloor (c + 1)K \rfloor}.\hspace{-.75mm}
\label{eq:j2jiayou11}
\end{eqnarray} 
Hence, we obtain a simpler version of Theorem \ref{thm:general_1}.
 
\vspace{1mm} 
\begin{theorem}\label{thm:general_3}
Let $\mathbf{x} \in \mathbb{R}^n$ be any $K$-sparse signal and let $\mathbf{\Phi} \in \mathbb{R}^{m \times n}$ be the measurement matrix satisfying the RIP of order $\lfloor (c + 1)K \rfloor$. Then if $c$ satisfies \eqref{eq:j2jiayou11}, OMP$_{cK}$ perfectly recovers the signal $\mathbf{x}$ from the measurements $\mathbf{y} = \mathbf{\Phi x}$.  
\end{theorem} 
\vspace{1mm}

\begin{figure}[t]
\begin{center}
\hspace{-1.5mm}\includegraphics[scale = .5]{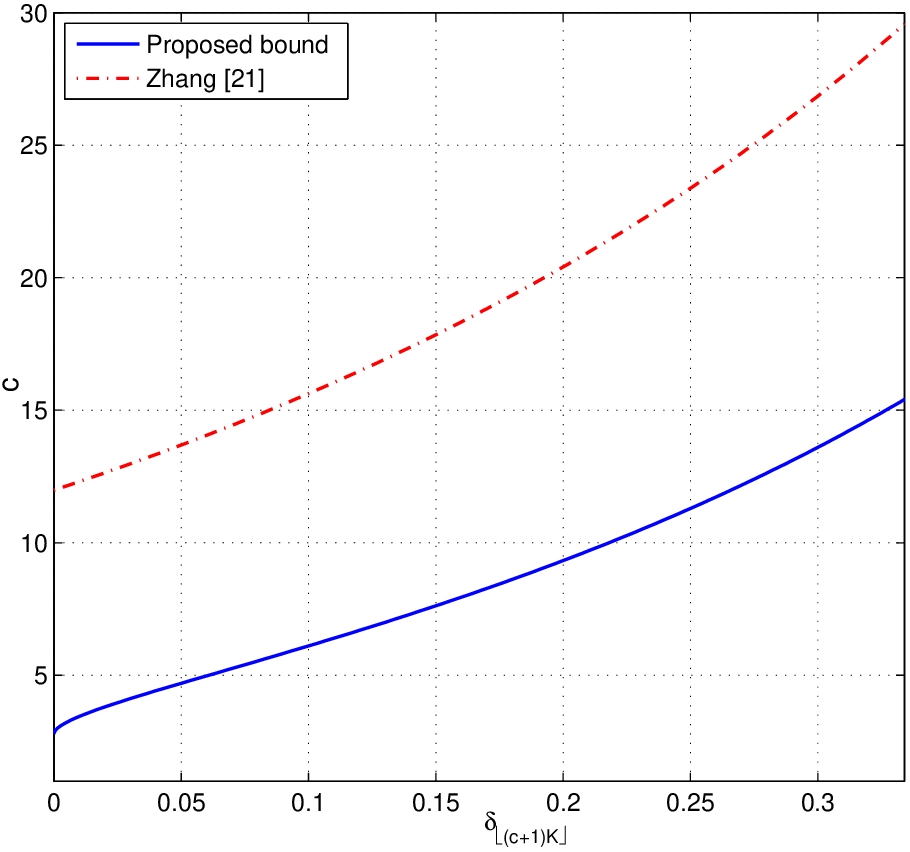} 
\caption{Comparison between the proposed bound and Zhang's result
\cite{zhang2011sparse}.}\label{fig:comparison}\vspace{-3mm}
\end{center}
\end{figure}

In Fig.~\ref{fig:comparison}, we compare the lower bound of $c$ in
\eqref{eq:j2jiayou11} with the result of Zhang~\cite[Theorem
2.1]{zhang2011sparse}.
In both cases, one can observe that the lower bound of $c$ increases with the RIC monotonically. In particular, the lower bound $c$ of this work is
uniformly smaller than that of~\cite{zhang2011sparse} for the whole range of RIC. 
For example, it requires $30K$ iterations to recover $K$-sparse
signals with $\delta_{31K} = \frac{1}{3}$ in~\cite{zhang2011sparse}. Whereas, it requires only $\lceil 15.4K \rceil$ iterations in our new
result.\footnote{From Fig.~\ref{fig:comparison}, the
RIP condition associated with $\lceil 15.4K \rceil$ iterations is
$\delta_{\lfloor 16.4 K \rfloor} \leq \frac{1}{3}$ in our result, which is less restrictive than $\delta_{31K} < \frac{1}{3}$~\cite{zhang2011sparse}.}

Another interesting point we observe from Fig.~\ref{fig:comparison} is the difference in the critical
value of $c$ such that $\delta_{\lfloor (s + 1)K \rfloor} = 0$.
In~\cite{zhang2011sparse}, the lower bound of $c$ ensuring $\delta_{\lfloor (s + 1)K \rfloor} > 0$ is 
\begin{equation}
c > 4 \log 20 \approx 12,
\end{equation}
while that in our work is
\begin{equation}
c > 4 \log 2 \approx 2.8.
\end{equation}
This means that if OMP$_{cK}$ runs at least $\lceil 2.8K \rceil$ iterations,
there always exists an absolute constant $\delta_{\lfloor (c + 1)K
\rfloor} \in (0, 1)$ satisfying \eqref{eq:j2jiayou11}, and under this condition OMP performs the exact recovery of all $K$-sparse signals. In fact, by applying $c = 2.8$ to \eqref{eq:j2jiayou11}, we obtain the upper bound of the RIC as 
\begin{equation}
\delta_{\lceil 3.8K \rceil} \leq 2 \cdot 10^{-5}.
\end{equation}


\begin{remark}[RIP condition] \label{rem:1h}
It is worth mentioning that the value of $c$ and the upper bound of the RIC are closely related. In fact, we can obtain better (larger) RIC bounds by using larger values of $c$. 
For example, when using $c = 30$ in Theorem~\ref{thm:general_3}, we can obtain the upper bound of the RIC as
\begin{equation}
\delta_{31K} \leq \frac{1}{2},
\end{equation} which is better (less restrictive) than the result $\delta_{31K} < \frac{1}{3}$~\cite{zhang2011sparse}.
\end{remark}
\vspace{1mm}

\begin{remark}[Why smaller $c$?] \label{rem:smallc}
Running fewer number of iterations of OMP$_{cK}$ offers many benefits. Noting that the number of indices chosen by OMP$_{cK}$ should not exceed the number of measurements (i.e., $\lceil cK \rceil \leq m$),\footnote{Otherwise the signal estimation step (i.e., the least squares (LS) projection) in the OMP algorithm cannot be performed.} a smaller value of $c$ directly leads to a wider range of sparsity $K$ when $m$ is fixed. On the other hand, when $K$ is fixed, running fewer number of iterations is also beneficial since larger $c$ may require larger number of measurements. Furthermore, identifying a small value of $c$ is of importance for the noisy case because running too many iterations will degrade the denoising performance of the algorithm~\cite{ding2013perturbation}. 

\end{remark}

\vspace{1mm} 
\begin{remark}[Fundamental limit of $c$] \label{rem:limit}
While Theorem \ref{thm:general_3} demonstrates that OMP$_{cK}$ can uniformly recover all $K$-sparse signals using at least $\lceil 2.8K \rceil$ iterations, it should be noted that the number cannot be smaller or equal to $K$, given that the RIC is an absolute constant. In fact, to ensure exact recovery with the conventional OMP algorithm (i.e., OMP$_{cK}$ with $c = 1$), the RIC should be at least inversely proportional to $\sqrt K$~\cite{wang2012Recovery,mo2012remarks}, which therefore places a fundamental limit to the recovery performance of OMP$_{cK}$. Our result bridges the gap between the result of Zhang~\cite{zhang2011sparse} and this fundamental limit.
\end{remark}

%
%

\vspace{1mm}

\begin{remark}[Measurement size]
It is well known that many random measurement matrices satisfy the
RIP with overwhelming probability when the number of measurements
scales linearly with the sparsity. For example, a random matrix
$\mathbf{\Phi} \in \mathbb{R}^{m \times n}$ with entries drawn
i.i.d. from Gaussian distribution $\mathcal{N}(0, \frac{1}{m})$
obeys the RIP with $\delta_K = \varepsilon \in (0, 1)$ with overwhelming
probability if $m \geq
\frac{ C K \log \frac{n}{K}}{\varepsilon^2}$ for some constant $C >
0$~\cite{candes2005decoding,baraniuk2008simple}. Therefore, our result implies that with high probability, OMP$_{cK}$
can recover $K$-sparse signals in $\lceil 2.8 K \rceil$ iterations
when the number of random Gaussian  measurements is on the order of
$K \log \frac{n}{K}$. This is essentially an encouraging result
since by running slightly more than the sparsity level $K$, OMP$_{cK}$ can
accurately reconstruct all $K$-sparse signals with the same order of
measurements as required by the state of the art sparse recovery
algorithms (e.g., Basis Pursuit~\cite{chen2001atomic} and
CoSaMP~\cite{needell2009cosamp}).
This is in contrast to the conventional OMP algorithm, for which the RIC should be at least inversely proportional to $\sqrt K$ in order to ensure exact recovery of all $K$-sparse signals~\cite{wang2012Recovery,mo2012remarks}, and hence the required number of measurements should be on the order of $K^2 \log \frac{n}{K}$.

\end{remark}

\vspace{1mm} 

\begin{remark}[Comparison with \cite{livshitz2014sparse}] \label{rem:compare_20}
Our result is closely related to the recent work of Livshitz and Temlyakov~\cite{livshitz2014sparse}. In their work, authors considered random sparse signals and showed that with high probability, these signals can be recovered with $\lceil (1 + \epsilon)K \rceil$ iterations of OMP$_{cK}$~\cite{livshitz2014sparse}. Our result has two key distinctions over this work. Firstly, while each nonzero component of sparse signals are upper bounded in~\cite{livshitz2014sparse}, our result does not impose any constraint on the nonzero components of input signals. Secondly, and more importantly, the analysis in~\cite{livshitz2014sparse} relies on the assumption that $K \geq \delta_{2K}^{-1/2}$, which essentially applies to the situation where the sparsity $K$ is nontrivial. In contrast, our analysis works for input signals with arbitrary sparsity levels. 
 
\end{remark}

\begin{figure}[t]
\centering \subfigure[Set diagram of $T$, $T^{k}$, and
$\Gamma^{k}$.]
{\includegraphics[scale = .875] {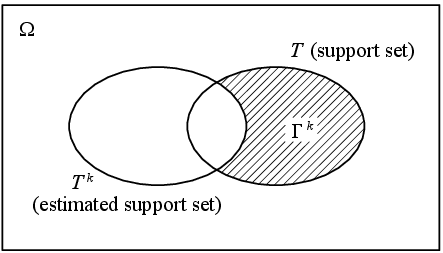}
\label{fig:subfig1}} \subfigure[Illustration of indices in
$\Gamma^{k}$.]
{\includegraphics[scale = .875] {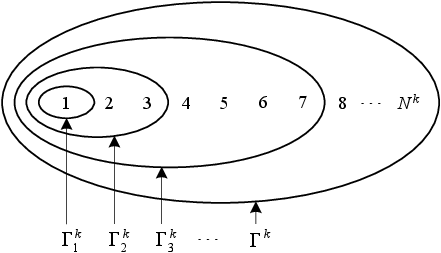}
\label{fig:subfig2}} \caption{Illustration of sets $T$, $T^{k}$, and
$\Gamma^{k}$.} \label{fig:set}\vspace{-2mm}
\end{figure}

\section{Proof of Theorem \ref{thm:general_1}} \label{sec:proof}

\subsection{Preliminaries} \label{sec:notations}

For notational simplicity, we denote $\Gamma^k = T \backslash T^{k}$
so that $N^k = |\Gamma^k|$. Also, without loss of generality, we
assume that $\Gamma^{k} = \{1, \cdots, N^{k}\}$ and that nonzero
elements in $\mathbf{x}_{\Gamma^k}$ are arranged in descending order
of their magnitudes (i.e., $|x_1| \geq \cdots \geq \left|
x_{N^k} \right|$). Now, we define the subset of
${\Gamma^{k}}$ as
\begin{equation}
{\Gamma}^k_{\tau} =
\begin{cases}
\emptyset                          & \tau = 0, \\
\{1, 2, \cdots, 2^{\tau} - 1\} & \tau = 1, 2, \cdots, \lfloor \log_2 {N^k} \rfloor, \\
\Gamma^{k}                         & \tau = \lfloor \log_2 {N^k}
\rfloor + 1.
\end{cases} \label{eq:jjjjffff}
\end{equation}
See Fig.~\ref{fig:set} for the illustration of
${\Gamma}^k_{\tau}$. Notice that the last set ${\Gamma}^k_{\lfloor
\log_2 {N^k} \rfloor + 1}$ $(= \Gamma^k)$ may have less than
$2^{\lfloor \log_2 {N^k} \rfloor + 1} - 1$ elements. For example, if
$\Gamma^{k} = \{1,2, \cdots, 10\}$, then ${\Gamma}^k_{0} =
\emptyset$, ${\Gamma}^k_1 = \{1\}$, ${\Gamma}^k_2 = \{1,2,3\}$,
${\Gamma}^k_3 = \{1, \cdots, 7\}$, and ${\Gamma}^k_4 = \{1,2,
\cdots, 10\}$ has less than $2^4 - 1$ $( = 15)$ elements.

For given set $\Gamma^{k}$ and constant $\sigma > 1$, let $L \in
\{1, \cdots, \lfloor \log_2 {N^k} \rfloor + 1\}$ be the minimum
integer satisfying
\begin{subequations}
\begin{align}
 \|\mathbf{x}_{{\Gamma^{k}} \backslash {\Gamma}^k_{0}}\|_2^2 &< \sigma
 \|\mathbf{x}_{{\Gamma^{k}} \backslash {\Gamma}^k_{1}}\|_2^2,
 \label{eq:mu1}\\
 \|\mathbf{x}_{{\Gamma^{k}} \backslash {\Gamma}^k_{1}} \|_2^2 &< \sigma
 \|\mathbf{x}_{{\Gamma^{k}} \backslash {\Gamma}^k_{2}}\|_2^2,
 \label{eq:mu2} \\
 &~~\vdots \nonumber \\
 \|\mathbf{x}_{{\Gamma^{k}} \backslash {\Gamma}^k_{{L}- 2}}\|_2^2 &< \sigma
 \|\mathbf{x}_{{\Gamma^{k}} \backslash {\Gamma}^k_{{L}- 1}}\|_2^2,
 \label{eq:mu3} \\
 \|\mathbf{x}_{{\Gamma^{k}} \backslash {\Gamma}^k_{{L}- 1}}\|_2^2 &\geq \sigma
 \|\mathbf{x}_{{\Gamma^{k}} \backslash {\Gamma}^k_{L}}\|_2^2.
 \label{eq:mu4}
\end{align}
\end{subequations}
Moreover, if \eqref{eq:mu4} holds true for all $L \geq 1$, then we
simply take $L = 1$ and ignore \eqref{eq:mu1}--\eqref{eq:mu3}. We
remark that $L$ always exists because $\|\mathbf{x}_{\Gamma^k
\backslash \Gamma^k_{\lfloor \log_2 {N^k} \rfloor + 1}}\|_2^2 = 0$
so that \eqref{eq:mu4} holds true at least for $L = \lfloor \log_2
{N^k} \rfloor + 1$. 


\vspace{-1mm}

\subsection{Main Idea}
The proof of Theorem \ref{thm:general_1} is based on
mathematical induction in $N^k$, the number of remaining support indices after $k$ iterations of OMP$_{cK}$. First, we check the case where $N^k
= 0$. This case is trivial since it implies that all support indices
have already been selected ($T \subseteq T^k$) so that no more iteration is needed. Then, we assume that the argument holds up to an
integer $\gamma - 1$ ($\gamma \geq 1$). In other words, we
assume that whenever $ N^k \leq \gamma - 1$, it requires at most $\lceil
c N^k \rceil$ additional iterations to select all remaining indices
in $T$. Under this inductive assumption, we will show that if $N^k
= \gamma$, it also requires at most $\lceil c \gamma \rceil$ additional
iterations to select the remaining $\gamma$ indices in $T$ (i.e., $T
\subseteq T^{k + \lceil c \gamma \rceil}$).

We now proceed to the proof of the induction step ($N^k \hspace{-.25mm} = \hspace{-.25mm} \gamma$).

\begin{itemize}
\item First, we show that after a specified number of additional iterations, a substantial amount of indices in $\Gamma^{k}$ can 
be chosen, and the number of remaining support indices is upper bounded. 
More precisely, we show that OMP$_{cK}$ chooses at least $2^{L - 1}$ support
indices of $\Gamma^k$ in
\begin{equation} \label{eq:jjff66}
 k' := \lceil c 2^{L - 1}\rceil - 1 
\end{equation}
additional iterations (where $L$ is defined in \eqref{eq:mu1}--\eqref{eq:mu4}) so that the number of remaining
support indices (after $k + k'$ iterations) satisfies
\begin{equation} \label{eq:jjff}
 N^{k + k'} \leq \gamma - 2^{{L}- 1}.
\end{equation}
\item Second, since (\ref{eq:jjff}) directly implies
$N^{k + k'} \leq \gamma - 1$, by induction hypothesis it
requires at most $\lceil c N^{k + k'} \rceil$ additional
iterations to choose the rest of support indices. In summary, the total number of
iterations of OMP$_{cK}$ to choose all support indices is no more than 
\begin{eqnarray}
 k \hspace{-.5mm} + \hspace{-.5mm} k' \hspace{-.5mm} + \hspace{-.5mm} \lceil c N^{k + k'} \rceil \hspace{-2mm}
 &\overset{(a)}{\leq}& \hspace{-2mm} k + k' + \lceil c ({\gamma} - 2^{{L}- 1}) \rceil \label{eq:jjff200} \nonumber \\
 \hspace{-2mm} &\overset{(b)}{=}& \hspace{-2mm} k \hspace{-.5mm} + \hspace{-.5mm} (\lceil c 2^{L - 1}\rceil \hspace{-.5mm} - \hspace{-.5mm} 1) \hspace{-.5mm} + \hspace{-.5mm} \lceil c ({\gamma} \hspace{-.5mm} - \hspace{-.5mm}
2^{{L}- 1}) \rceil \label{eq:dengdaiyiii4444} \nonumber \\
 \hspace{-2mm} &\overset{(c)}{\leq}& \hspace{-2mm} k + \lceil c \gamma \rceil, \label{eq:dengdaiyiii}
\end{eqnarray}
where (a) is from \eqref{eq:jjff}, (b) follows from \eqref{eq:jjff66}, and (c) holds true because $\lceil a \rceil + \lceil
b \rceil - 1 \leq \lceil a + b \rceil$. 

\end{itemize}

Therefore, we can conclude that all support
indices are chosen in $k + \lceil
c \gamma \rceil$ iterations of OMP$_{cK}$ ($T \subseteq T^{k + \lceil c \gamma
\rceil}$), which establishes the induction step.
An illustration of the induction step is given in
Fig.~\ref{fig:idea}. 
Now, what remains is the proof of~\eqref{eq:jjff}.

%

\vspace{-1mm}

\subsection{Sketch of Proof for (\ref{eq:jjff})} \label{sec:proof_of_13}
Before we proceed, we explain the key idea to prove this claim. Instead of proving \eqref{eq:jjff} directly, we show that a sufficient condition of \eqref{eq:jjff} is true. That is,
\begin{equation}
\| \mathbf{x}_{\Gamma^{k + k'}} \|_2^2 < \|\mathbf{x}_{\Gamma^{k}
\backslash {\Gamma}^k_{{L}- 1}} \|_2^2. \label{eq:32hh}
\end{equation}

We first explain why \eqref{eq:32hh} is a sufficient condition of~\eqref{eq:jjff}. 
Consider $\mathbf{x}_{\Gamma^{k + k'}}$ and
$\mathbf{x}_{\Gamma^{k} \backslash {\Gamma}^k_{{L}- 1}}$, which are
two truncated vectors of $\mathbf{x}_{\Gamma^{k}}$.
 From the definition of $\Gamma^k_{\tau}$
in \eqref{eq:jjjjffff}, we have ${{\Gamma}^k_{L - 1}} = {\{1, 2,
\cdots, 2^{L - 1} - 1\}}$ and $ |\Gamma^{k} \backslash
{\Gamma}^k_{{L}- 1}| = |\{2^{L - 1}, \cdots, \gamma\}| = \gamma -
2^{L - 1} + 1$. Thus, we can obtain an alternative form of \eqref{eq:jjff} as
$
N^{k + k'} \leq |\Gamma^k \backslash \Gamma^k_{L - 1}| - 1.$
Equivalently,
\begin{equation} \label{eq:jjffabc}
N^{k + k'} < |\Gamma^k \backslash \Gamma^k_{L - 1}|.
\end{equation}
Further, since $|x_i|$, $i = 1, 2, \cdots, \gamma$, are
arranged in descending order of their magnitudes (i.e., $ |x_1| \geq \cdots \geq |x_{\gamma}|$),
$\mathbf{x}_{\Gamma^{k} \backslash {\Gamma}^k_{{L}- 1}} = (x_{2^{L -
1}}, \cdots, x_{\gamma})'$ consists of ${\gamma} - 2^{{L}- 1} + 1$
smallest elements (in magnitude) of $\mathbf{x}_{\Gamma^{k}}$.
Then it is easy to see that \eqref{eq:32hh}
ensures \eqref{eq:jjffabc}, and thus \eqref{eq:32hh} becomes a sufficient condition of~\eqref{eq:jjff}.

\begin{figure}[t]
\begin{center}
\hspace{-1.5mm} \includegraphics[scale = .63]{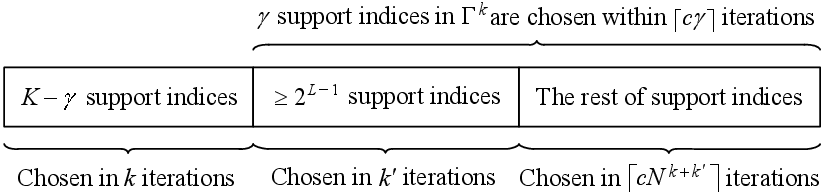}
\caption{Illustration of the induction step when $N^k =
\gamma$.}\label{fig:idea} \vspace{-5mm}
\end{center}
\end{figure}


%
%
%

Now, what remains is the proof of \eqref{eq:32hh}. To the end, we build an upper bound for $\| \mathbf{x}_{\Gamma^{k +
k'}} \|_2^2$ and a lower bound for $\|
\mathbf{x}_{\Gamma^{k} \backslash {\Gamma}^k_{{L}- 1}} \|_2^2$, and then relate them to get a condition for~\eqref{eq:32hh}.
 
\vspace{1mm}
 
\begin{proposition} \label{prop:bound}
We have
\begin{eqnarray}
\| \mathbf{x}_{\Gamma^{k +
k'}} \|_2^2 &\leq& \frac{\| \mathbf{r}^{k + k'} \|_2^2}{1 - \delta_{|T \cup
T^{k + \lfloor c \gamma \rfloor}|}}, \\
\|\mathbf{x}_{ \Gamma^{k} \backslash {\Gamma}^k_{{L}- 1}} \|_2^2 &>& \frac{\sigma(1 - \sigma  \eta) \| \mathbf{r}^{k + k'} \|_2^2} {(1 + \delta_{\gamma}) (1 - \eta)},
\end{eqnarray} 
where $\eta := \exp \Big(- \frac{c (1 - \delta_{|T \cup
T^{k + \lfloor c \gamma \rfloor}|})}{4 (1 + \delta_1)}
 \Big)$.
\end{proposition}

\vspace{1mm}
The proof is left to Appendix \ref{app:bound}.
%
From Proposition~\ref{prop:bound}, it is clear that \eqref{eq:32hh} holds true whenever  
\begin{equation} \label{eq:27hh}
\frac{1}{1 - \delta_{|T \cup
T^{k + \lfloor c \gamma \rfloor}|}} \leq \frac{\sigma(1 - \sigma  \eta) } {(1 + \delta_{\gamma}) (1 - \eta)}.
\end{equation}
%
Noting that $ \sigma (1 - \sigma \eta) =
\frac{1}{\eta} ( {-(\sigma \eta - \frac{1}{2})^2 + \frac{1}{4}})$, by choosing $\sigma \eta = \frac{1}{2}$ we have $\sigma (1 -
\sigma \eta) = \frac{1}{4 \eta}$, and hence \eqref{eq:27hh} becomes
\begin{equation}
4  \eta (1 - \eta) (1 + \delta_{\gamma}) \leq 1 - \delta_{|T \cup T^{k + \lfloor c \gamma \rfloor}|}.
\end{equation}
Equivalently,
\begin{equation}
 c \hspace{-.5mm} \geq \hspace{-.5mm} - \frac{4 (1 \hspace{-.5mm} + \hspace{-.5mm} \delta_1)}{1 - \delta_{s}} \hspace{-.5mm}
  \log \hspace{-.5mm} \left( \hspace{-.5mm} \frac{1}{2} \hspace{-.75mm} - \hspace{-.75mm} \frac{1}{2}  \sqrt{\frac{\delta_{\gamma} \hspace{-.5mm} + \hspace{-.5mm} \delta_{s}}{1 + \delta_{\gamma}}} \right)\text{where}~s \hspace{-.5mm} := \hspace{-.75mm} |T \cup T^{k + \lfloor c \gamma \rfloor}|.  \nonumber
\end{equation}
Thus completes the proof.

\section{Conclusion}
\label{sec:conclusion}

In this paper, we have investigated the recovery performance of OMP when the number of iterations exceeds the sparsity $K$ of input signals. We have established a lower bound on the
number of iterations of OMP, expressed in terms of RIC, that guarantees the exact recovery of sparse signals. Our result demonstrates that OMP can accurately
recover any $K$-sparse signal in $\lceil 2.8K \rceil$ iterations with the RIC being an absolute constant. This result bridges the gap between the recent result of~\cite{zhang2011sparse} and the fundamental limit of the OMP algorithm at which exact sparse recovery cannot be uniformly ensured.
%
%
Considering that a large number of iterations leads to a high
computational complexity and also imposes a strict limitation on the
sparsity level ($K \leq \frac{m}{c}$), the reduction on the number
of iterations offers computational benefits as well as
relaxations in the measurement size and the sparsity range of underlying signals to be recovered.

\appendices
\numberwithin{equation}{section}
\newcounter{mytempthcnt}
\setcounter{mytempthcnt}{\value{theorem}}
\setcounter{theorem}{2}

\section{Proof of Proposition \ref{prop:bound}} \label{app:bound}
\subsection{Upper bound for $\|
\mathbf{x}_{\Gamma^{k + k'}} \|_2^2$}
 
Since $\mathbf{r}^{k + k'} = \mathbf{y}
  - \mathbf{\Phi} \hat{\mathbf{x}}^{k + k'} = \mathbf{\Phi} ( \mathbf{x}
  - \hat{\mathbf{x}}^{k + k'} )$, 
\begin{eqnarray}
  \| \mathbf{r}^{k + k'} \|_2^2 &\overset{(a)}{\geq}& ({1 - \delta_{|T \cup T^{k + k'}|}} ) \| \mathbf{x} - \hat{\mathbf{x}}^{k + k'}
  \|_2^2  \nonumber
\\
  &\overset{(b)}{\geq}&   ({1 - \delta_{|T \cup T^{k + k'}|}} ) \|\mathbf{x}_{\Gamma^{k + k'}}
  \|_2^2 \nonumber \\
  &\overset{(c)}{\geq}&  ({1 - \delta_{|T \cup T^{k + \lfloor c \gamma \rfloor}|}} ) \|\mathbf{x}_{\Gamma^{k + k'}}
  \|_2^2,
\end{eqnarray}
where (a) is from the RIP, (b) is due to $\| \mathbf{x} - \hat{\mathbf{x}}^{k + k'}  \|_2^2 \geq \|
\mathbf{x}_{T \backslash T^{k + k'}} \|_2^2 = \|
\mathbf{x}_{\Gamma^{k + k'}} \|_2^2$, and (c) is because \begin{equation}
k + k' \overset{(d)}{=} k + \lceil c2^{L - 1} \rceil - 1 \overset{(e)}{\leq} k + \lceil c \gamma \rceil - 1 \leq k + \lfloor c \gamma \rfloor, \label{eq:fact1}
\end{equation} 
where (d) is from \eqref{eq:jjff66} and (e) is because $2^{L - 1} \leq \gamma$.

\subsection{Lower bound for $\|
\mathbf{x}_{\Gamma^{k} \backslash {\Gamma}^k_{{L}- 1}} \|_2^2$} 
Compared to the upper bound analysis, the lower bound analysis
requires a little more effort. The following lemmas will be used in our analysis.


%
\vspace{1mm}

\begin{lemma} \label{lem:rleq}
The residual $\mathbf{r}^k$ satisfies
\begin{equation}
\| \mathbf{r}^{k} \|_2^2 \leq (1 + \delta_{N^k})
\|\mathbf{x}_{\Gamma^{k}}\|_2^2.
\end{equation}
\end{lemma}

\begin{proof}
From Table~\ref{tab:OMP}, the residual of OMP can be expressed as
$\mathbf{r}^{k} = \mathbf{y} - \mathbf{\Phi} \hat{\mathbf{x}}^k = \mathbf{y} - \mathcal{P}_{T^k} \mathbf{y}$. Since $\mathbf{r}^{k} \bot \mathcal{P}_{T^k} \mathbf{y}$,
\begin{eqnarray}
  \|\mathbf{r}^{k}\|_2^2 &=& \|\mathbf{y}\|_2^2 - \|\mathcal{P}_{T^k} \mathbf{y}\|_2^2 \nonumber \\
  &\overset{(a)}{\leq}& \|\mathbf{y}\|_2^2 - \|\mathcal{P}_{T^k
\cap T} \mathbf{y}\|_2^2 \nonumber \\
  &\overset{(b)}{=}&  \|\mathbf{y} - \mathcal{P}_{T^k \cap T}
\mathbf{y}\|_2^2,\label{eq:chacha1}
\end{eqnarray}
where (a) is due to $(T^{k} \cap T) \subseteq T^{k}$ so that
${span}(\mathbf{\Phi}_{T^{k}}) \supseteq {span}(\mathbf{\Phi}_{T^{k}
\cap T})$ and $\|\mathcal{P}_{T^k} \mathbf{y}\|_2^2 \geq |\mathcal{P}_{T^k \cap T} \mathbf{y}\|_2^2$ and (b) is because $(\mathbf{y} - \mathcal{P}_{T^k \cap T}
\mathbf{y}) \bot \mathcal{P}_{T^k \cap T} \mathbf{y}$.
Noting that $\mathcal{P}_{T^{k} \cap T} \mathbf{y}$ is the
projection of $\mathbf{y}$ onto ${span}(\mathbf{\Phi}_{T^{k} \cap
T})$, we have
\begin{eqnarray}
\|\mathbf{y} - \mathcal{P}_{T^{k} \cap T} \mathbf{y} \|_2^2 &=&
\min_{\mathbf{u}:{supp}(\mathbf{u}) = T^{k} \cap T}
{\|\mathbf{y}-\mathbf{\Phi} \mathbf{u}\|}_2^2 \nonumber \\
&\leq&{\|\mathbf{y}-\mathbf{\Phi}_{T^{k} \cap T}
\mathbf{x}_{T^{k} \cap T} \|}_2^2 = {\|\mathbf{\Phi}_{\Gamma^{k}} \mathbf{x}_{\Gamma^{k}} \|}_2^2 \nonumber \\
&\leq& (1 + \delta_{N^k}) \|\mathbf{x}_{\Gamma^{k}}\|_2^2,
\label{eq:ggsssss1} \nonumber
\end{eqnarray}
where the last inequality is from the RIP ($|\Gamma^{k}| = N^k$).

This, together with \eqref{eq:chacha1}, establishes the lemma.
\end{proof}

\vspace{1mm}


The second lemma provides a lower bound for $\| \mathbf{r}^l \|_2^2 - \| \mathbf{r}^{l + 1} \|_2^2$
in the $(l + 1)$-th ($l \geq k$) iteration of OMP.

\vspace{1mm}

\begin{lemma} \label{prop:residual1}
For given integer $l \geq k$, we have
\begin{equation} \label{eq:38hh}
 \|\mathbf{r}^l \|_2^2 - \| \mathbf{r}^{l + 1} \|_2^2 \geq  \frac{1 - \delta_{ | {\Gamma}^k_{\tau} \cup T^l |}}{(1 +
\delta_1) | {\Gamma}^k_{\tau} |}  \left( \| \mathbf{r}^l \|_2^2 -
  \| \mathbf{\Phi}_{{\Gamma^{k}} \backslash {\Gamma}^k_\tau} \mathbf{x}_{{\Gamma^{k}} \backslash {\Gamma}^k_\tau} \|_2^2
  \right),
\end{equation}
where $\tau = 1, \cdots, \lfloor \log_2 {N^k} \rfloor + 1$.
\end{lemma}

\vspace{1mm}
\begin{proof}
The proof consists of two parts. First, we show that the residual
power difference of OMP satisfies
\begin{equation} \label{eq:residual_A111}
  \|\mathbf{r}^l\|_2^2 - \|\mathbf{r}^{l + 1}\|_2^2 \geq \frac{\|\mathbf{\Phi}' \mathbf{r}^l\|_\infty^2}{1 + \delta_{1}}.
\end{equation}
Second, we show that
  \begin{equation} \label{eq:a111}
\hspace{-2.5mm}\|\mathbf{\Phi}' \mathbf{r}^l \|_\infty^2 \geq  \frac{1 - \delta_{ |
{\Gamma}^k_{\tau} \cup T^l |}}{ | {\Gamma}^k_{\tau}
  |}  \left( \| \mathbf{r}^l \|_2^2 -
  \| \mathbf{\Phi}_{{\Gamma^{k}} \backslash {\Gamma}^k_\tau} \mathbf{x}_{{\Gamma^{k}} \backslash {\Gamma}^k_\tau} \|_2^2
  \right).
  \end{equation}
The lemma is established by combining \eqref{eq:residual_A111}
and \eqref{eq:a111}.
  
\vspace{1mm}  
  
{\it 1) Proof of \eqref{eq:residual_A111}}: 
  Observe that the residual of OMP satisfies
\begin{eqnarray}
 \mathbf{r}^l - \mathbf{r}^{l + 1} 
&=& (\mathbf{y} - \mathcal{P}_{T^{l}} \mathbf{y}) - (\mathbf{y} -
\mathcal{P}_{T^{l + 1}} \mathbf{y}) \nonumber \\
&\overset{(a)}{=}& (\mathcal{P}_{T^{l + 1}} - \mathcal{P}_{T^{l + 1}}
\mathcal{P}_{T^{l}} )\mathbf{y} \nonumber \\ &=& \mathcal{P}_{T^{l +
1}} (\mathbf{y} - \mathcal{P}_{T^{l}}
\mathbf{y}) = \mathcal{P}_{T^{l + 1}} \mathbf{r}^l, \label{eq:44c}
\end{eqnarray}
where (a) is because ${span}(\mathbf{\Phi}_{T^{l}})
\subseteq {span}(\mathbf{\Phi}_{T^{l + 1}})$ so that
$\mathcal{P}_{T^l} \mathbf{y} = \mathcal{P}_{T^{l + 1}}
\mathcal{P}_{T^{l}} \mathbf{y}$.
Recalling that $t^{l + 1}$ is the index chosen at the $(l + 1)$-th
iteration and $T^{l + 1} = T^l \cup t^{l + 1}$, we have
${span}(\mathbf{\Phi}_{T^{l + 1}}) \supseteq {span}({\phi}_{t^{l +
1}})$, and hence
\begin{eqnarray} \label{eq:residual_A}
  \|\mathbf{r}^l - \mathbf{r}^{l + 1}\|_2^2 = \|\mathcal{P}_{T^{l + 1}}
 \mathbf{r}^l\|_2^2 \geq \|\mathcal{P}_{t^{l + 1}}
 \mathbf{r}^l\|_2^2. \nonumber
\end{eqnarray}
Further, noting that $\|\mathbf{r}^l - \mathbf{r}^{l + 1}\|_2^2 =
\|\mathbf{r}^l\|_2^2 - \|\mathbf{r}^{l + 1}\|_2^2$ and that
$\mathcal{P}_{t^{l + 1}} = \mathcal{P}'_{t^{l + 1}} = (
\mathbf{\phi}_{t^{l + 1}}^\dag)' \mathbf{\phi}'_{t^{l + 1}}$, we
have
\begin{eqnarray}
  \|\mathbf{r}^l \|_2^2 - \|\mathbf{r}^{l + 1}\|_2^2 &\geq& \|\mathcal{P}_{t^{l + 1}}
  \mathbf{r}^l\|_2^2 = \|({\phi}_{t^{l + 1}}^\dag)' {\phi}'_{t^{l + 1}} \mathbf{r}^l\|_2^2 \nonumber \\
  &\overset{(a)}{=}& ({\phi}'_{t^{l + 1}} \mathbf{r}^l)^2 \|\mathbf{\phi}_{t^{l +
  1}}^\dag\|_2^2 \nonumber \\
  & \overset{(b)}{\geq} & \frac{({\phi}'_{t^{l + 1}} \mathbf{r}^l)^2 \|\mathbf{\phi}_{t^{l + 1}}^\dag \mathbf{\phi}_{t^{l + 1}}\|_2^2}{\|\mathbf{\phi}_{t^{l + 1}}\|_2^2} \nonumber \\
  & \overset{(c)}{\geq} & \frac{({\phi}'_{t^{l + 1}} \mathbf{r}^l)^2}{1 + \delta_1} = \frac{\|\mathbf{\Phi}' \mathbf{r}^l\|_\infty^2}{1 + \delta_{1}},
  \label{eq:mmmss12}
\end{eqnarray}
where (a) is because ${\phi}'_{t^{l + 1}}
\mathbf{r}^l$ is a scalar, (b) is from the norm inequality, and (c) is due to the RIP.

\begin{figure}[t]
\begin{center}
\includegraphics[scale = 1]{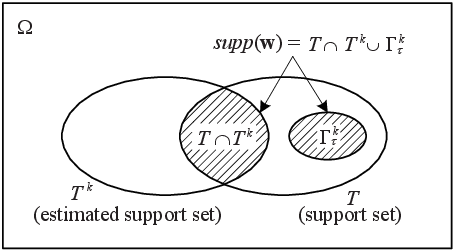}
\caption{Illustration of ${supp}(\mathbf{w})$.} \label{fig:set_3}\vspace{-2mm}
\end{center}
\end{figure}

\vspace{1mm}

 {\it 2) Proof of \eqref{eq:a111}}: First, let $\mathbf{w} \in
\mathbb{R}^n$ be the vector such that
\begin{equation} \label{eq:jjjjffff0000}
\mathbf{w}_S =
\begin{cases}
\mathbf{x}_S & S = T \cap T^{k} \cup
{\Gamma}^k_\tau,\\
\mathbf{0}         & S = \Omega \backslash (T\cap T^{k} \cup
{\Gamma}^k_\tau),
\end{cases}
\end{equation}
where $\tau \in \{ 1, 2,
\cdots, \lfloor \log_2 {N^k} \rfloor + 1\}$.
An illustration of ${supp}(\mathbf{w})$ is provided in
Fig.~\ref{fig:set_3}. Then, by noting that 
${supp}(\mathbf{\Phi}' \mathbf{r}^l) = \Omega \backslash
T^l$, we have
%
 \begin{eqnarray}
 \langle \mathbf{\Phi}' \mathbf{r}^l, \mathbf{w} \rangle
&=& \langle (\mathbf{\Phi}' \mathbf{r}^l)_{\Omega \backslash T^l},
\mathbf{w}_{\Omega \backslash T^l} \rangle  \nonumber \\
&\overset{(a)}{\leq}& \|(\mathbf{\Phi}' \mathbf{r}^l)_{\Omega \backslash
T^l}\|_\infty \|\mathbf{w}_{\Omega \backslash T^l}\|_1 \nonumber  \\
&\overset{(b)}{\leq}& \sqrt{|\Omega \backslash T^l
|}~ \|\mathbf{\Phi}' \mathbf{r}^l \|_\infty
\|\mathbf{w}_{\Omega \backslash T^l}\|_2 \nonumber \\
&\overset{(c)}{\leq}& \sqrt{|{{\Gamma}^k_\tau}
|}~ \|\mathbf{\Phi}' \mathbf{r}^l \|_\infty
\|\mathbf{w}_{\Omega \backslash T^l}\|_2, \label{eq:youmyou}
 \end{eqnarray}
where (a) is from H\"{o}lder's inequality, 
(b) follows from the norm inequality
($\|\mathbf{u}\|_1 \leq \sqrt{\|\mathbf{u}\|_0} 
\|\mathbf{u}\|_2$), and (c) is because some indices in $\Gamma^k_\tau$ may be identified in iterations $k + 1, \cdots, l$ so that $|\Omega \backslash T^l| \leq |{{\Gamma}^k_\tau}
|$. Since ${supp}(\hat{\mathbf{x}}^l) = T^l$ and
${supp}(\mathbf{\Phi}' \mathbf{r}^l) = \Omega \backslash T^l$, it is
clear that $\langle \mathbf{\Phi}' \mathbf{r}^l, \hat{\mathbf{x}}^l
\rangle = \mathbf{0}$ and $\langle \mathbf{\Phi}' \mathbf{r}^l,
\mathbf{w} \rangle = \langle \mathbf{\Phi}' \mathbf{r}^l, \mathbf{w}
-\hat{\mathbf{x}}^l \rangle$. Thus, \eqref{eq:youmyou} can be
rewritten as
 \begin{eqnarray} \label{eq:33333}
\|\mathbf{\Phi}' \mathbf{r}^l \|_\infty \geq \frac{\langle
\mathbf{\Phi}' \mathbf{r}^l, \mathbf{w} \rangle}
{\sqrt{|{\Gamma}^k_\tau|} ~ \|\mathbf{w}_{\Omega \backslash T^l}\|_2} =
\frac{\langle \mathbf{\Phi}' \mathbf{r}^l, \mathbf{w}
-\hat{\mathbf{x}}^l \rangle}{\sqrt{|{\Gamma}^k_\tau |}~
\|\mathbf{w}_{\Omega \backslash T^l}\|_2}.
 \end{eqnarray}

Next, we build a lower bound for $\langle \mathbf{\Phi}' \mathbf{r}^l, \mathbf{w} - \hat{\mathbf{x}}^l \rangle$. Note that
\begin{eqnarray}
 \lefteqn{2 \langle \mathbf{\Phi}' \mathbf{r}^l, \mathbf{w} -
\hat{\mathbf{x}}^l \rangle = 2 \langle \mathbf{\Phi}(\mathbf{w} - \hat{\mathbf{x}}^l), \mathbf{r}^l \rangle } \nonumber \\
 &=& \|\mathbf{\Phi}(\mathbf{w} - \hat{\mathbf{x}}^l)\|_2^2 + \|\mathbf{r}^l\|_2^2 - \|\mathbf{r}^l - \mathbf{\Phi}(\mathbf{w} - \hat{\mathbf{x}}^l)\|_2^2
\nonumber \\
 &\overset{(a)}{=}& \|\mathbf{\Phi}(\mathbf{w} - \hat{\mathbf{x}}^l)\|_2^2 + \|\mathbf{r}^l\|_2^2  - \|\mathbf{\Phi}(\mathbf{x} - \mathbf{w})\|_2^2 \label{eq:geisiche99} \nonumber \\
 &=& \|\mathbf{\Phi}(\mathbf{w} - \hat{\mathbf{x}}^l)\|_2^2 + \|\mathbf{r}^l\|_2^2 - \|\mathbf{\Phi}_{{\Gamma^{k}}\backslash {\Gamma}^k_\tau} \mathbf{x}_{{\Gamma^{k}} \backslash {\Gamma}^k_\tau}\|_2^2, ~~~~~~~\label{eq:101}\end{eqnarray}
where (a) holds because $\mathbf{r}^l +
\mathbf{\Phi} \hat{\mathbf{x}}^l = \mathbf{y} = \mathbf{\Phi}
\mathbf{x}$.
When $\|\mathbf{r}^l\|_2^2 - \|\mathbf{\Phi}_{{\Gamma^{k}}
\backslash {\Gamma}^k_\tau} \mathbf{x}_{{\Gamma^{k}} \backslash
{\Gamma}^k_\tau}\|_2^2 \geq 0$,\footnote{We only need to consider $\|\mathbf{r}^l\|_2^2 - \|\mathbf{\Phi}_{{\Gamma^{k}}\backslash {\Gamma}^k_\tau} \mathbf{x}_{{\Gamma^{k}} \backslash
{\Gamma}^k_\tau}\|_2^2 \geq 0$ because if $\|\mathbf{r}^l\|_2^2 - \|\mathbf{\Phi}_{{\Gamma^{k}}\backslash {\Gamma}^k_\tau} \mathbf{x}_{{\Gamma^{k}} \backslash
{\Gamma}^k_\tau}\|_2^2 < 0$, (\ref{eq:a111}) holds trivially since 
$\|\mathbf{\Phi}' \mathbf{r}^l \|_\infty^2 \geq 0$.}
  using $a^2 + b^2
\geq 2 {ab}$ in (\ref{eq:101}) yields
\begin{equation}
 \langle \mathbf{\Phi}' \mathbf{r}^l, \mathbf{w} - \hat{\mathbf{x}}^l \rangle \geq  \|\mathbf{\Phi}(\mathbf{w} - \hat{\mathbf{x}}^l) \|_2
\sqrt{\|\mathbf{r}^l \|_2^2 - \|\mathbf{\Phi}_{{\Gamma^{k}}
\backslash {\Gamma}^k_\tau} \mathbf{x}_{{\Gamma^{k}} \backslash
{\Gamma}^k_\tau}\|_2^2}.  \nonumber
\end{equation}
Moreover, since ${supp}(\mathbf{w} - \hat{\mathbf{x}}^l) = (T \cap T^{k}
\cup {\Gamma}^k_\tau) \cup T^l \subseteq {\Gamma}^k_{\tau} \cup
T^l$, 
  \begin{eqnarray}
  \|\mathbf{\Phi}(\mathbf{w} - \hat{\mathbf{x}}^l) \|_2 &\overset{(a)}{\geq}& \sqrt{1 -
  \delta_{|{\Gamma}^k_{\tau} \cup T^l|}} ~ \|\mathbf{w} -
  \hat{\mathbf{x}}^l\|_2 \nonumber \\
  &\geq& \sqrt{1 - \delta_{|{\Gamma}^k_{\tau} \cup T^l|}} ~ \|(\mathbf{w} - \hat{\mathbf{x}}^l)_{\Omega \backslash
  T^l}\|_2 \nonumber \\
  \label{eq:103}&\overset{(b)}{=}& \sqrt{1 - \delta_{|{\Gamma}^k_{\tau} \cup T^l|}} ~
  \|\mathbf{w}_{\Omega \backslash T^l}\|_2,
  \end{eqnarray}
where (a) is from the RIP and (b) is uses $(\hat{\mathbf{x}}^l)_{\Omega
\backslash T^l} = \mathbf{0}$. Hence, 
\begin{eqnarray}
 \lefteqn{\langle \mathbf{\Phi}' \mathbf{r}^l, \mathbf{w} - \hat{\mathbf{x}}^l \rangle \geq \|\mathbf{w}_{\Omega \backslash T^l}\|_2} \nonumber\\
 &~~~
\times  \sqrt{(1 - \delta_{|{\Gamma}^k_{\tau} \cup T^l|}) (\|\mathbf{r}^l \|_2^2 - \|\mathbf{\Phi}_{{\Gamma^{k}}
\backslash {\Gamma}^k_\tau} \mathbf{x}_{{\Gamma^{k}} \backslash
{\Gamma}^k_\tau}\|_2^2)}.~~~~ \label{eq:dada222}
\end{eqnarray}

Finally, plugging 
(\ref{eq:dada222}) into (\ref{eq:33333}), we obtain (\ref{eq:a111}).
\end{proof}
 
\vspace{1mm}
The consequence of Lemma~\ref{prop:residual1} is the following lemma, which is crucial for proving the lower bound for $\|
\mathbf{x}_{\Gamma^{k} \backslash {\Gamma}^k_{{L}- 1}} \|_2^2$.
 
\vspace{1mm}

\begin{lemma} \label{lem:residual}
For any integer $l' \geq l \geq k$ and $\tau \in  \{1, \cdots, \lfloor
\log_2 {N^k} \rfloor + 1\}$, the residual of OMP satisfies
\begin{eqnarray}
\hspace{-1mm} \| \mathbf{r}^{l'} \hspace{-.5mm} \|_2^2 \hspace{-.5mm} - \hspace{-.75mm} \| \mathbf{\Phi}_{{\Gamma^{k}}
\backslash
  {\Gamma}^k_\tau} \mathbf{x}_{{\Gamma^{k}} \backslash
  {\Gamma}^k_\tau}\|_2^2 \hspace{-.5mm} \leq \hspace{-.5mm} C_{\tau,l, l'} \hspace{-.5mm} \left(\hspace{-.25mm}  \| \mathbf{r}^l \hspace{-.25mm} \|_2^2 \hspace{-.5mm} - \hspace{-.75mm} \| \mathbf{\Phi}_{{\Gamma^{k}} \backslash
  {\Gamma}^k_\tau} \mathbf{x}_{{\Gamma^{k}} \backslash
  {\Gamma}^k_\tau}\|_2^2\right) \nonumber
\end{eqnarray}
where $C_{\tau,l, l'} = \exp \Big( - \frac{(1 - \delta_{| {\Gamma}^k_\tau
\cup T^{l' - 1}|}) (l' - l)} {(1 + \delta_1) |
{\Gamma}^k_\tau |} \Big)$.
\end{lemma} 
\vspace{1mm}
\begin{proof}
Using $a \geq 1 - e^{-a}$ with $a = \frac{1 - \delta_{|
{\Gamma}^k_\tau \cup T^l
 |}}{(1 + \delta_1) | {\Gamma}^k_\tau  |}$, we have
\begin{equation} \label{eq:45hh}
{ \frac{1 - \delta_{ |
{\Gamma}^k_\tau \cup T^l
 |}}{(1 + \delta_1) | {\Gamma}^k_\tau  |} \geq 1 - \exp  \left( - \frac{1 -
\delta_{ | {\Gamma}^k_\tau \cup T^l  |}}{(1 + \delta_1)
|{\Gamma}^k_\tau  |}
 \right)} > 0.
\end{equation}
Then, we can rewrite \eqref{eq:38hh} in Lemma~\ref{prop:residual1} as\footnote{\label{positive}Note that $\|\mathbf{r}^l \|_2^2 - \| \mathbf{r}^{l + 1} \|_2^2 \geq 0$ due to orthogonal projection at each iteration of OMP. When $\| \mathbf{r}^l \|_2^2 -
  \| \mathbf{\Phi}_{{\Gamma^{k}} \backslash {\Gamma}^k_\tau} \mathbf{x}_{{\Gamma^{k}} \backslash {\Gamma}^k_\tau} \|_2^2 \geq 0$, \eqref{eq:residualabc} directly follows from \eqref{eq:38hh} and \eqref{eq:45hh}. When $ \| \mathbf{r}^l \|_2^2 -
  \| \mathbf{\Phi}_{{\Gamma^{k}} \backslash {\Gamma}^k_\tau} \mathbf{x}_{{\Gamma^{k}} \backslash {\Gamma}^k_\tau} \|_2^2 < 0$, \eqref{eq:residualabc} holds trivially because in this case its right-hand side is negative. }
\begin{eqnarray}
 && \hspace{-1mm} \|\mathbf{r}^l \|_2^2 - \| \mathbf{r}^{l + 1} \|_2^2 \geq  \left(1 - \exp  \left( - \frac{1 -
\delta_{ | {\Gamma}^k_\tau \cup T^l  |}}{(1 + \delta_1)
|{\Gamma}^k_\tau  |} \right) \right) \nonumber \\
&&~~~~~~~~~~~~~~~~~~\times \left( \| \mathbf{r}^l \|_2^2 -
  \| \mathbf{\Phi}_{{\Gamma^{k}} \backslash {\Gamma}^k_\tau} \mathbf{x}_{{\Gamma^{k}} \backslash {\Gamma}^k_\tau} \|_2^2
  \right).~~~~~~~~~ \label{eq:residualabc}
\end{eqnarray}
Subtracting both sides of (\ref{eq:residualabc}) by $\| \mathbf{r}^l
\|_2^2 - \| \mathbf{\Phi}_{T \backslash {\Gamma}^k_\tau}
\mathbf{x}_{T \backslash {\Gamma}^k_\tau} \|_2^2$, we have
\begin{eqnarray}
&& \hspace{-1mm} \| \mathbf{r}^{l + 1}  \|_2^2 -  \| \mathbf{\Phi}_{{\Gamma^{k}}
\backslash {\Gamma}^k_\tau} \mathbf{x}_{{\Gamma^{k}} \backslash
{\Gamma}^k_\tau}\|_2^2 \hspace{1.5mm} \leq \hspace{1.5mm} \exp \left( - \frac{1 -
\delta_{|{\Gamma}^k_\tau \cup T^l|}}{(1 + \delta_1) |{\Gamma}^k_\tau
|}\right) \nonumber \\
&&~~~~~~~~~~~~~~~~~~\times (\| \mathbf{r}^l  \|_2^2 - \|
\mathbf{\Phi}_{{\Gamma^{k}} \backslash {\Gamma}^k_\tau}
\mathbf{x}_{{\Gamma^{k}} \backslash {\Gamma}^k_\tau} \|_2^2),
\label{eq:A1}
\end{eqnarray}
and thus
\begin{eqnarray}
&& \hspace{-1mm}\| \mathbf{r}^{l + 2}  \|_2^2 -  \| \mathbf{\Phi}_{{\Gamma^{k}}
\backslash {\Gamma}^k_\tau} \mathbf{x}_{{\Gamma^{k}} \backslash
{\Gamma}^k_\tau} \|_2^2 \leq\exp \left( - \frac{1 -
\delta_{|{\Gamma}^k_\tau \cup T^{l + 1}|}} {(1 + \delta_1) |
{\Gamma}^k_\tau |}\right) \nonumber \\
&&~~~~~~~~~~~~~~~~~~\times (\|\mathbf{r}^{l + 1}  \|_2^2 -  \| \mathbf{\Phi}_{{\Gamma^{k}} \backslash {\Gamma}^k_\tau} \mathbf{x}_{{\Gamma^{k}} \backslash {\Gamma}^k_\tau}  \|_2^2),~~ \\
&&~~~~~~~~~~~~~\vdots \nonumber  \\
&& \hspace{-1mm}\| \mathbf{r}^{l'}  \|_2^2 - \| \mathbf{\Phi}_{{\Gamma^{k}}
\backslash {\Gamma}^k_\tau} \mathbf{x}_{{\Gamma^{k}} \backslash
{\Gamma}^k_\tau} \|_2^2 \hspace{1mm} \leq  \hspace{1mm} \exp \left( - \frac{1 - \delta_{ |
{\Gamma}^k_\tau \cup T^{l' - 1}|}} {(1 + \delta_1) |{\Gamma}^k_\tau
|}\right)
 \nonumber \\
&&~~~~~~~~~~~~~~~~~~\times\| \mathbf{r}^{l' - 1}  \|_2^2 - { \|
\mathbf{\Phi}_{{\Gamma^{k}} \backslash {\Gamma}^k_\tau}
\mathbf{x}_{{\Gamma^{k}} \backslash {\Gamma}^k_\tau} \|_2^2}).~~
\label{eq:B}
\end{eqnarray}
From \eqref{eq:A1}--\eqref{eq:B}, we have
\begin{eqnarray}
 \lefteqn{\| \mathbf{r}^{l'} \|_2^2 -  \| \mathbf{\Phi}_{{\Gamma^{k}}
\backslash {\Gamma}^k_\tau} \mathbf{x}_{{\Gamma^{k}} \backslash
{\Gamma}^k_\tau} \|_2^2} \nonumber \\
 &\leq& \hspace{-2mm} \prod_{i = l}^{l' - 1} \exp \left( - \frac{1 -
\delta_{|{\Gamma}^k_\tau \cup T^{i}|}}{(1 + \delta_1) |
{\Gamma}^k_\tau |}\right) (\| \mathbf{r}^{l}  \|_2^2 - {
\|\mathbf{\Phi}_{{\Gamma^{k}} \backslash {\Gamma}^k_\tau}
\mathbf{x}_{{\Gamma^{k}} \backslash {\Gamma}^k_\tau} \|_2^2}) \nonumber \\
&\leq& \hspace{-2mm} C_{\tau,l, l'} (\| \mathbf{r}^l \|_2^2 - { \|
\mathbf{\Phi}_{{\Gamma^{k}} \backslash {\Gamma}^k_\tau}
\mathbf{x}_{{\Gamma^{k}} \backslash {\Gamma}^k_\tau} \|_2^2}),
\nonumber
\end{eqnarray}
where
$C_{\tau,l, l'} = \exp \left( - \frac{(1 - \delta_{| {\Gamma}^k_\tau
\cup T^{l'- 1}|}) (l' - l)} {(1 + \delta_1) | {\Gamma}^k_\tau |}
\right)$.
 \end{proof}

\vspace{2mm}

Now, we are ready to identify the lower bound for $\|
\mathbf{x}_{\Gamma^{k} \backslash {\Gamma}^k_{{L}- 1}} \|_2^2$. First, let
$k_0 = k$ and $
k_i = k + \sum_{\tau = 1}^i \left\lceil \frac{c}{4}
|{\Gamma}^k_{\tau}| \right\rceil$ for $i = 1, \cdots,
{L},$
where ${\Gamma}^k_{\tau}$ and $L$ are defined in (\ref{eq:jjjjffff})
and~\eqref{eq:mu1}--\eqref{eq:mu4}, respectively. Then, by applying
Lemma \ref{lem:residual} with $l' = k_i$ and $l = k_{i - 1}$,
$i = 1, \cdots, L$, we have
\begin{eqnarray}
  \lefteqn{\| {{\mathbf{r}}^{k_i}} \|_{2}^{2} - \| \mathbf{\Phi}_{{\Gamma^{k}}
\backslash {\Gamma}^k_i} \mathbf{x}_{{\Gamma^{k}} \backslash
  {\Gamma}^k_i}\|_2^2} \nonumber \\ &~\leq
  C_{i,k_{i-1}, k_i} (\|  \mathbf{r}^{k_{i-1}} \|_2^2 - \| \mathbf{\Phi}_{{\Gamma^{k}} \backslash
  {\Gamma}^k_{i}} \mathbf{x}_{{\Gamma^{k}} \backslash
  {\Gamma}^k_i}\|_2^2), \label{eq:11f}
\end{eqnarray}
where $
 C_{i,k_{i - 1},k_i} = \exp \left( - \frac{(1 - \delta_{| {\Gamma}^k_i \cup T^{k_i
- 1}|}) (k_{i} - k_{i - 1}) } {(1 + \delta_1) |
{\Gamma}^k_i |}  \right)$.

We next build an upper bound for the constant $C_{i,k_{i - 1},k_i}$ in~\eqref{eq:11f}. Since  $\frac{k_i - k_{i - 1}}{|{\Gamma}^k_{i}|} = \frac{\lceil \frac{c}{4} |{\Gamma}^k_i | \rceil}{|{\Gamma}^k_{i}|}  \geq
 \frac{c}{4}$ for $i = 1, 2, \cdots, L$, and also noting that $k_i - 1 \leq k_L$,
we can rewrite $C_{i,k_{i - 1},k_i}$ as
\begin{eqnarray}
  C_{i, k_{i - 1}, k_i}  \leq \exp \left(- \frac{c (1 - \delta_{|T \cup T^{k_{L}}|})}{4 (1 + \delta_1)} \right), \label{eq:61hh}
\end{eqnarray}
where we have used monotonicity of the RIC ($|{\Gamma}^k_i \cup T^{k_{i} - 1}| \leq |{\Gamma}^k \cup
T^{k_{i} - 1}| \leq |T \cup T^{k_L}|$ for $i = 1, \cdots, L$).

Further, we find an upper bound for $k_L$ in~\eqref{eq:61hh}. Recalling from \eqref{eq:jjjjffff} that $|{\Gamma}^k_\tau| \leq 2^{\tau} - 1$ for
$\tau = 1, 2, \cdots, L$, we have
\begin{equation} \label{eq:62hh}
 k_{L} = k + \sum_{\tau = 1}^L  \left\lceil \frac{c}{4} |{\Gamma}^k_\tau| \right\rceil \leq k + \sum_{\tau = 1}^{L} \left\lceil \frac{c}{4} (2^{\tau} - 1) \right\rceil. 
\end{equation}
Since \eqref{eq:j2jiayou111} directly implies that $c \geq 4 \log 2 \geq 2$ under the RIP assumption in Theorem~\ref{thm:general_1}, one can show that (see Appendix~\ref{app:d})
\begin{equation}
 \sum_{\tau = 1}^{L} \left\lceil \frac{c}{4} (2^{\tau} - 1) \right\rceil \leq \lceil c 2^{L - 1} \rceil - 1, \label{eq:61h}
 \end{equation}
 Hence, \eqref{eq:62hh} becomes
\begin{equation}
 k_{L} \leq k + \lceil c 2^{L - 1} \rceil - 1  = k + k' \overset{(a)}{\leq} k + \lfloor c \gamma \rfloor,  \label{eq:64h}
\end{equation}
where (a) is from~\eqref{eq:fact1}. This, together with \eqref{eq:61hh}, implies 
\begin{eqnarray}
 C_{i, k_{i - 1}, k_i} \leq \exp \left(- \frac{c (1 - \delta_{|T \cup T^{k + \lfloor c \gamma \rfloor}|})}{4 (1 + \delta_1)}
 \right). \label{eq:chh}
\end{eqnarray}

%
%
%
Now we can construct an upper bound for $\| \mathbf{r}^{k_i} \|_2^2$ using~\eqref{eq:11f} and~\eqref{eq:chh}. By denoting $\eta := \exp \left(- \frac{c (1 - \delta_{|T \cup T^{k + \lfloor c \gamma \rfloor}|})}{4 (1 + \delta_1)}
 \right),$ we rewrite (\ref{eq:11f}) as
\begin{equation}
\label{eq:11ff}
  \| \mathbf{r}^{k_i} \|_2^2
  \leq \eta \| \mathbf{r}^{k_{i-1}} \|_2^2 + (1 - \eta) \| \mathbf{\Phi}_{{\Gamma^{k}} \backslash
  {\Gamma}^k_i} \mathbf{x}_{{\Gamma^{k}} \backslash  {\Gamma}^k_i}\|_2^2, ~i = 1, \cdots, L. \nonumber
\end{equation}
Some additional manipulations yield the following result,
\begin{eqnarray}
  \lefteqn{\hspace{-.5mm}\| \mathbf{r}^{k_{L}} \|_2^2} \nonumber \\
  &\hspace{-5mm} \leq& \hspace{-4mm} \eta^{L} \| \mathbf{r}^{k} \|_2^2 + (1 - \eta) \sum_{\tau = 1}^{L}\eta^{{L}- \tau} \|\mathbf{\Phi}_{{\Gamma^{k}} \backslash {\Gamma}^k_\tau}  \mathbf{x}_{{\Gamma^{k}} \backslash {\Gamma}^k_\tau} \|_2^2 \nonumber \\
  &\hspace{-5mm}\overset{(a)}{\leq}& \hspace{-4mm}  \eta^{L} \| \mathbf{r}^{k} \|_2^2 + (1 - \eta) (1 + \delta_\gamma) \sum_{\tau = 1}^{L}\eta^{{L}-
  \tau} \|\mathbf{x}_{{\Gamma^{k}} \backslash {\Gamma}^k_\tau} \|_2^2 \nonumber
   \\
  &\hspace{-5mm}\overset{(b)}{\leq}&  \hspace{-4mm} \left(\eta^L \|\mathbf{x}_{\Gamma^{k} \backslash {\Gamma}^{k}_0}\|_2^2 + (1 - \eta)  \sum_{\tau = 1}^{L}\eta^{{L}-
  \tau} \| \mathbf{x}_{{\Gamma^{k}} \backslash {\Gamma}^k_\tau}\|_2^2\right) (1 + \delta_{\gamma}) \nonumber
   \\
  &\hspace{-5mm}\overset{(c)}{\leq}& \hspace{-4mm}  \left({(\sigma  \eta )^{{L}}} + (1 - \eta)  \sum_{\tau = 1}^{L} (\sigma  \eta )^{L - \tau} \right)  \frac{(1 + \delta_{\gamma}) \| \mathbf{x}_{{\Gamma^{k}} \backslash {\Gamma}^k_{{L}- 1}} \|_2^2}{\sigma}    \nonumber \\
  &\hspace{-5mm}\overset{(d)}{<}& \hspace{-4mm}  \left( \hspace{-.5mm} (1\hspace{-.5mm} - \hspace{-.5mm}\eta)\hspace{-.5mm} \sum_{\tau = L}^{\infty} (
  \sigma  \eta )^{\tau} \hspace{-1mm} + (1 \hspace{-.5mm} -\hspace{-.5mm} \eta)\hspace{-.5mm} \sum_{\tau = 0}^{L - 1} (
  \sigma  \eta )^{\tau} \hspace{-.75mm} \right) \hspace{-1mm} \frac{(1 \hspace{-.75mm} + \hspace{-.5mm} \delta_{\gamma}) \| \mathbf{x}_{{\Gamma^{k}} \backslash {\Gamma}^k_{{L}- 1}} \hspace{-.5mm} \|_2^2}{\sigma} \nonumber \\
  &\hspace{-5mm} = & \hspace{-4mm}   \frac{(1 + \delta_{\gamma}) (1 - \eta) \| \mathbf{x}_{{\Gamma^{k}} \backslash {\Gamma}^k_{{L}- 1}}
  \|_2^2}{\sigma(1 - \sigma \eta)},
\end{eqnarray}
where (a) is due to the RIP ($|{\Gamma^{k}} \backslash
{\Gamma}^k_\tau| \leq |\Gamma^k| = \gamma$ for $\tau = 1, \cdots,
L$), (b) is from Lemma \ref{lem:rleq} ($\|
\mathbf{r}^{k} \|_2^2 \leq (1 + \delta_{N^k})
\|\mathbf{x}_{\Gamma^{k}}\|_2^2 = (1 + \delta_{\gamma})
\|\mathbf{x}_{\Gamma^{k} \backslash {\Gamma}^{k}_0}\|_2^2$),
(c) follows from
\begin{equation}
  \| \mathbf{x}_{{\Gamma^{k}} \backslash {\Gamma}^k_\tau} \|_2^2 \leq \sigma ^{{L}- 1 - \tau} \|
\mathbf{x}_{{\Gamma^{k}} \backslash {\Gamma}^k_{{L}- 1}}
  \|_2^2,  \label{eq:next2} ~~~  \tau = 0, 1, \cdots, {L}, \nonumber
\end{equation}
which is a direct consequence of \eqref{eq:mu1}--\eqref{eq:mu4},
and (d) is because
 $\sigma
> 1$ and $\eta < 1$ so that $(\sigma  \eta )^{L} < \frac{1 - \eta}{1 - \sigma \eta} \cdot (\sigma
\eta)^{L} = (1 - \eta) \sum_{\tau = L}^{\infty} (\sigma
\eta)^{\tau}$  when $0 < \sigma \eta < 1$.
 
Finally, since $k_L \leq k + k'$ by \eqref{eq:64h}, and also noting that
the residual power of OMP is always non-increasing, we have
\begin{equation}
  \| \mathbf{r}^{k + k'} \|_2^2 \leq   \| \mathbf{r}^{k_{L}} \|_2^2
  \leq \frac{(1 + \delta_{\gamma}) (1 - \eta) \| \mathbf{x}_{{\Gamma^{k}} \backslash {\Gamma}^k_{{L}- 1}}
\|_2^2}{\sigma(1 - \sigma
\eta)} , \label{eq:43045} \nonumber
\end{equation}
which is the desired result.

%

\section{Proof of \eqref{eq:61h}} \label{app:d}

\begin{proof}
We prove \eqref{eq:61h} by mathematical induction on $L$. 
First, when $L = 1$, \eqref{eq:61h} becomes $
\left \lceil \frac{c}{4} \right \rceil \leq \lceil c \rceil - 1,$
which is simply true since $c \geq 2$.
Next, we assume that \eqref{eq:61h} holds up to an integer $\ell$ 
 so that $
 \sum_{\tau = 1}^\ell \left \lceil \frac{c}{4} (2^\tau - 1) \right \rceil \leq \left \lceil c 2^{\ell - 1} \right \rceil - 1$.
Then if $L = \ell + 1$,
\begin{eqnarray}
 \sum_{\tau = 1}^{\ell + 1} \left \lceil \frac{c}{4} (2^\tau - 1) \right \rceil 
&\leq&
   \left \lceil c 2^{\ell - 1} \right \rceil - 1 + \left \lceil \frac{c}{4} (2^{\ell + 1} - 1) \right \rceil
  \nonumber \\
  &\overset{(a)}{\leq}&
   \left \lceil c 2^{\ell - 1} \right \rceil + \left \lceil c 2^{\ell - 1} - \frac{1}{2} \right \rceil - 1
  \nonumber \\
  &\overset{(b)}{\leq}&
   \left \lceil c 2^{\ell} \right \rceil - 1,
\end{eqnarray} 
where (a) is because $c \geq 2$ and (b) uses Hermite's identity~\cite{patashnik1989concrete} ($\lceil ax \rceil = \lceil x \rceil + \lceil x - \frac{1}{a} \rceil \cdots + \lceil x - \frac{a - 1}{a} \rceil$ with $a = 2$ and $x = c 2^{\ell - 1}$), which completes the proof.
\end{proof}

\bibliographystyle{IEEEbib}
\bibliography{CS_refs}

\end{document}